\documentclass[pdflatex,sn-mathphys-num]{sn-jnl}

\usepackage{amsmath,amssymb,amsthm}
\usepackage{bm}
\usepackage{graphicx}
\graphicspath{{figures/}}
\usepackage{booktabs}
\usepackage{float}
\usepackage{placeins}
\usepackage{enumitem}
\usepackage{tabularx}
\usepackage{xcolor}
\usepackage{tikz}
\usetikzlibrary{arrows.meta,positioning,shapes.geometric}
\usepackage[nameinlink,noabbrev]{cleveref}

\theoremstyle{plain}
\newtheorem{theorem}{Theorem}
\newtheorem{proposition}{Proposition}
\newtheorem{lemma}{Lemma}

\theoremstyle{definition}

\newtheorem{assumption}{Assumption}
\theoremstyle{remark}
\newtheorem{remark}{Remark}

\newcommand{\R}{\mathbb{R}}
\newcommand{\E}{\mathbb{E}}

\newcommand{\ip}[2]{\left\langle #1,#2 \right\rangle}

\newcommand{\VWK}{\textrm{VWK}}
\newcommand{\He}{\mathrm{He}}
\newcommand{\Var}{\operatorname{Var}}
\newcommand{\Span}{\operatorname{span}}

\begin{document}

\title[Skew Penalty for Volterra Cross-Correlation Identification]{A Closed-Form Skew Penalty for Volterra Cross-Correlation Identification under Non-Gaussian Input}

\author[1,2,3]{\fnm{Serhii} \sur{Zabolotnii}}\email{zabolotnii.serhii@csbc.edu.ua}

\affil[1]{\orgname{Cherkasy State Business College}, \city{Cherkasy}, \postcode{18028}, \country{Ukraine}}
\affil[2]{\orgname{State Scientific Research Institute of Armament and Military Equipment Testing and Certification}, \city{Cherkasy}, \country{Ukraine}}
\affil[3]{\orgname{Uzhhorod National University}, \city{Uzhhorod}, \country{Ukraine}}

\abstract{The Lee--Schetzen cross-correlation method identifies finite-memory Volterra kernels by projecting onto a Wiener--Hermite basis matched to Gaussian white input; off that input the diagonal estimator is biased, but the size of the bias has not been given in closed form. This paper's main result is an exact order-2 misspecification-penalty theorem: expressed in the variance-matched Gaussian basis, the cross-correlation estimator of a finite-memory Volterra functional incurs an excess $L^2(P)$ risk governed entirely by the input skew coefficient $\delta=\mu_3/\sigma^2$, and this penalty vanishes if and only if the input law is symmetric. For a centered-exponential input the classical estimator carries $32\times$ the risk of the estimator matched to the input law. The remedy is to identify the kernels in the coordinate system that is orthonormal for the actual input law $P$: the distribution-matched Volterra--Wiener--Kunchenko (\VWK{}) basis, built by oriented Gram--Schmidt on the monomials in $L^2(P)$ and available in closed form from empirical moments. This basis is the identification-from-data specialization of arbitrary polynomial chaos \citep{Xiu2002,OladyshkinNowak2012}; its Gaussian row recovers the Wiener--Hermite basis and its discrete and bounded rows recover the Askey families as closed-form checks. A conditioning analysis shows that the monomial Gram condition number exceeds $10^6$ by degree five and $10^9$ by degree six, while the matched Gram operator is the identity at every order; a machine-checked Lean~4 proof establishes the Binomial$(N,p)$ Krawtchouk row for arbitrary $N$. The analysis is moment-based, finite-memory, and restricted to product input laws.}


\maketitle

\section{Introduction}
\label{sec:introduction}

Finite-order Volterra series are a classical coordinate system for nonlinear
systems, representing the output as a polynomial functional of past inputs
\citep{Volterra1930,Schetzen1980,Marmarelis2004}; in a discrete finite-memory
setting this is a polynomial in a lag vector
$(x_t,x_{t-1},\ldots,x_{t-d+1})$. The algebra is elementary, but the geometry is
not: monomial features are correlated under the input law, so coefficient
identification depends on the inner product that law induces.

Wiener's construction resolves this in one setting. For Gaussian white input,
monomial Volterra functionals orthogonalize into Hermite functionals, which
yields the Wiener--Hermite expansion and a practical cross-correlation kernel
estimator \citep{Wiener1958,LeeSchetzen1965}. The limitation is not orthogonality
but the fixed Gaussian geometry. Real input laws may be skewed, discrete, or
platykurtic, and then the natural coordinate system is orthogonal in $L^2(P)$ for
the actual input law $P$, not in $L^2(N(0,\sigma^2))$.

This paper takes the matched coordinate system---the orthonormal polynomial basis
generated by $P$ in $L^2(P)$---as a means of identifying finite-memory Volterra
kernels from data, where $P$ enters only through empirical moments. Three results
are new, and we state them first.
\begin{enumerate}[leftmargin=*]
\item A misspecification-penalty theorem (\cref{prop:penalty}) that quantifies, in
closed form, the excess risk of the variance-matched Gaussian/Wiener
cross-correlation estimator under an asymmetric input law, governed by the skew
coefficient $\delta=\mu_3/\sigma^2$.
\item A finite-memory Volterra-kernel identification framing---the coefficient
bijection between monomial and matched coordinates (\cref{thm:mapping}) and the
conditioning advantage of the matched basis over the raw power basis
(\cref{subsec:conditioning})---that specializes arbitrary Polynomial Chaos
\citep{OladyshkinNowak2012,Xiu2002} to identification from data with empirical
moments.
\item A machine-checked Lean~4 proof (\cref{prop:krawtchouk}) of the
Binomial$(N,p)\to$Krawtchouk orthogonality for arbitrary $N$, resting on a third
Bernstein factorial-moment lemma absent from Mathlib.
\end{enumerate}

The matched basis itself is not new: the moment $\to$ oriented Gram--Schmidt
$\to$ distribution-matched orthonormal construction is arbitrary Polynomial Chaos
\citep{OladyshkinNowak2012}. Its classical reductions are recovery checks rather
than contributions, the Gaussian row recovering the Wiener--Hermite basis
\citep{Wiener1958,CameronMartin1947} and the discrete and bounded laws recovering
the Askey families \citep{Xiu2002}.

The analysis is deliberately moment-based and assumes a finite-memory product
input law $P^d$; it does not treat continuous-time integral kernels, correlated
or dependent lag laws, or a characteristic-function / moment-free formulation
for inputs without finite moments. The last is deferred to the next paper.

\section{Related work}
\label{sec:related}

\paragraph{Volterra and Wiener--Hermite identification.}
Volterra series provide a functional expansion for nonlinear systems
\citep{Volterra1930}. Wiener specialized the orthogonalization problem to
Gaussian random input, obtaining the Hermite-functional representation
\citep{Wiener1958}. Lee and Schetzen gave a practical cross-correlation procedure
for estimating Wiener kernels \citep{LeeSchetzen1965}, and later treatments
developed the theory and applications of Volterra/Wiener models
\citep{Schetzen1980,Marmarelis2004}. VWK keeps the same orthogonal-projection
motivation but replaces the Gaussian inner product by $L^2(P)$.

\paragraph{Orthogonal-polynomial and orthogonal Volterra identification.}
Removing the white-Gaussian-input requirement of Wiener--Hermite identification is not
new in itself: the fast orthogonal algorithm orthogonalizes the candidate terms on the
sample, so that arbitrary, including non-Gaussian, inputs are admissible without an
analytic weight \citep{Korenberg1988}. Orthonormal Laguerre and Kautz expansions instead
orthogonalize the memory axis, compressing long Volterra kernels into a few dynamic basis
functions \citep{Campello2004}. The recoverability and conditioning of Volterra and
polynomial regression under general, non-white inputs have been analyzed directly
\citep{KekatosGiannakis2011}, and input-distribution-matched orthogonal polynomial
nonlinear filters have been developed for adaptive filtering
\citep{Carini2015,RobustOrthBasis2014}, with a modern survey in \citep{Cheng2017}.
The demarcation from this line is at the level of the \emph{result}, not the motivation.
The distribution-matched adaptive filters of \citet{Carini2015,RobustOrthBasis2014} are built
precisely to remove input-mismatch, so the qualitative fact that a matched basis restores an
unbiased diagonal estimator is already implicit there; what is absent is a closed form for the
price of \emph{not} matching. The contribution here is that quantitative statement:
\cref{prop:penalty} gives the exact excess $L^2(P)$ risk of the mismatched diagonal
(cross-correlation) estimator in terms of the input moments $(\lambda,\rho,\delta)$, together with
the Askey reductions of the basis, the linear bijection between monomial Volterra coefficients
$h_\beta$ and VWK coordinates $a_\alpha$ on the shared index set $\mathcal{I}_{d,s}$, and the
conditioning theory that ties the orthogonalization to a quantified estimation advantage rather
than to numerical convenience alone.

\paragraph{Orthogonal polynomials and the Askey scheme.}
Classical orthogonal polynomial families can be characterized by the measures
with respect to which they are orthogonal. The Askey scheme organizes many of
these families and their limit relations \citep{Koekoek2010}. The VWK view is
not a new family competing with Askey polynomials. Rather, it is a
distribution-indexed construction whose closed forms coincide with the
appropriate classical family when such a family is available.

\paragraph{Generalized and arbitrary polynomial chaos.}
The distribution-matched orthonormal basis $\psi_0,\ldots,\psi_s$ of \cref{lem:vwk-basis}
is generalized polynomial chaos \citep{GhanemSpanos1991} in its Wiener--Askey form
\citep{Xiu2002}: each classical input law indexes the orthonormal polynomial family for
which it is the weight, and the Gaussian/Hermite, Binomial/Krawtchouk, Poisson/Charlier,
negative-binomial/Meixner, Gamma/Laguerre, and Beta/Jacobi rows are exactly that
correspondence. Constructing the basis numerically from the moment (Hankel) matrix $H_s$
for a law without a named weight is arbitrary polynomial chaos
\citep{OladyshkinNowak2012,SoizeGhanem2004,WanKarniadakis2006}, and its convergence and
regression-based estimation from data are well studied
\citep{Ernst2012,BlatmanSudret2011,Torre2019}. We therefore claim no novelty for the
basis itself, its uniqueness, or the Askey identifications, which we recover as
consistency checks. The relevant boundary for this paper is different: the gPC and
aPC literature usually starts from the matched basis, whereas the Volterra--Wiener
identification problem also requires explaining what the correlation estimator loses
when it is left in the mismatched Gaussian basis.

\paragraph{Higher-order statistics.}
Higher-order spectra use cumulants in the frequency domain for nonlinear signal
analysis and system identification \citep{Brillinger1965,Brillinger1975,Nikias1993}.
This tradition is closely related in its use of higher-order information, but it
has a different geometry. HOS identifies structure spectrally; VWK constructs
time-domain orthogonal coordinates in a finite polynomial space. The two views
are complementary rather than mutually exclusive.

\paragraph{Kunchenko stochastic-polynomial school.}
Kunchenko's stochastic-polynomial framework introduced polynomial approximation
in a space with a generating element and the normal system $F K=B$
\citep{Kunchenko2002,Kunchenko2006,Zabolotnii2026Survey}. Later PMM work developed low-order
non-Gaussian estimation for regression and time series
\citep{Zabolotnii2018,Zabolotnii2026}. In the present paper, this apparatus is
used as a distribution-matched orthogonalization principle for finite-memory
Volterra polynomials. The PMM connection is interpretive: low-order skew
corrections in the VWK basis explain why PMM2-like terms naturally arise in
asymmetric regimes, but this paper does not reduce all PMM estimators to VWK.

\paragraph{GMM and empirical characteristic functions.}
Generalized method of moments and empirical-characteristic-function estimation
provide another route from moment conditions to estimators
\citep{Hansen1982,FeuervergerMcDunnough1981,CarrascoFlorens2000,Yu2004}. This
paper does not use the CF/ECF machinery as a theorem. That connection requires
a separate characteristic-function construction in which finite raw moments are
not assumed.

\section{Mathematical setting}
\label{sec:setting}

Let $P$ be a probability law on $\R$, let $X\sim P$, and define
\[
  \ip{f}{g}_P = \E_P[f(X)g(X)].
\]
For a nonnegative integer $s$, let
\[
  \Pi_s=\Span\{1,x,\ldots,x^s\}\subset L^2(P).
\]
For memory length $d$, define the total-degree finite-memory polynomial space
\[
  \Pi_{d,s}
  =\Span\{x^\beta: \beta\in\mathbb{N}^d,\ |\beta|\le s\}
  \subset L^2(P^d),
\]
where $x^\beta=\prod_{r=1}^d x_r^{\beta_r}$ and
$|\beta|=\beta_1+\cdots+\beta_d$.

\begin{assumption}[Moment-based finite-memory regime]
\label{ass:moment}
Fix $s\ge 0$ and $d\ge 1$.
\begin{enumerate}[label=A\arabic*.,leftmargin=*]
\item Raw moments $m_r=\E[X^r]$ exist for $r=0,\ldots,2s$, with $m_0=1$.
\item $\Var(X)>0$.
\item The monomials $\{1,x,\ldots,x^s\}$ are linearly independent in $L^2(P)$.
\item For the finite-memory Volterra model, the input lag vector has product law
$P^d$.
\end{enumerate}
\end{assumption}

Equivalently, the Hankel moment matrix
$H_s=(m_{i+j})_{0\le i,j\le s}$ is positive definite. Indeed, for any
coefficient vector $c$,
\[
  c^\top H_s c = \E\left[\left(\sum_{i=0}^s c_i X^i\right)^2\right].
\]
For finite support of size $N+1$, the independence condition is equivalent to
$s\le N$.

\begin{table}[t]
\centering
\caption{Notation. The input law $P$ is centered, $\E[X]=0$, with $\sigma^2>0$.}
\label{tab:notation}
\begin{tabular}{@{}l p{0.60\columnwidth}@{}}
\toprule
Symbol & Meaning \\
\midrule
$P,\ X\sim P$ & input law on $\R$ and a draw from it \\
$\ip{f}{g}_P$ & inner product $\E_P[f(X)g(X)]$ on $L^2(P)$ \\
$\sigma^2$ & input variance $\Var(X)>0$ \\
$\mu_3,\ \mu_4$ & central moments $\E[X^3],\ \E[X^4]$ \\
$\delta$ & skew coefficient $\mu_3/\sigma^2$ \\
$\lambda$ & $\mu_3/\sigma$ \\
$\rho$ & $\rho^2=\mu_4-\sigma^4-\mu_3^2/\sigma^2\ge 0$ \\
$\psi_k$ & matched \VWK{} basis: oriented orthonormal polynomials of $P$; $\psi_0=1$, $\psi_1(x)=x/\sigma$, $\psi_2=\pi_2/\rho$ with $\pi_2(x)=x^2-\delta x-\sigma^2$ and $\lVert\pi_2\rVert_P^2=\rho^2$ \\
$g_k$ & variance-matched Wiener basis $\He_k(x/\sigma)/\sqrt{k!}$, orthonormal in $N(0,\sigma^2)$ but not in $P$ \\
$\Pi_s$ & univariate polynomials of degree $\le s$, $=\Span\{1,x,\dots,x^s\}$ \\
$\Pi_{d,s}$ & $d$-variate polynomials of total degree $\le s$ \\
$\mathcal{I}_{d,s}$ & multi-index set $\{\alpha\in\mathbb{N}^d:\lvert\alpha\rvert\le s\}$ \\
$\Psi_\alpha$ & tensor feature $\prod_{r=1}^{d}\psi_{\alpha_r}$, $\alpha\in \mathcal{I}_{d,s}$ \\
$a_\alpha$ & \VWK{} coordinate of $V$ on $\Psi_\alpha$ \\
$h_\beta$ & monomial coordinate of $V$ on $x^\beta$ \\
$H_s$ & Hankel/moment Gram matrix $(m_{i+j})_{0\le i,j\le s}$, $m_k=\E[X^k]$ \\
$T_s$ & lower-triangular change of basis, $T_s H_s T_s^\top=I$ \\
\bottomrule
\end{tabular}
\end{table}

\section{The VWK basis theorem}
\label{sec:basis}

\begin{lemma}[Oriented Gram--Schmidt VWK basis]
\label{lem:vwk-basis}
Under \cref{ass:moment} A1--A3, applying Gram--Schmidt to
$\{1,x,\ldots,x^s\}$ in $L^2(P)$ constructs a unique polynomial family
$\{\psi_0,\ldots,\psi_s\}$ such that:
\begin{enumerate}[label=(\roman*),leftmargin=*]
\item $\deg\psi_k=k$;
\item $\ip{\psi_i}{\psi_j}_P=\delta_{ij}$;
\item the leading coefficient of $\psi_k$ is positive;
\item the coefficients of $\psi_k$ depend only on $m_0,\ldots,m_{2k}$.
\end{enumerate}
\end{lemma}

\begin{proof}
Set $\psi_0=1$. Suppose $\psi_0,\ldots,\psi_{k-1}$ have been constructed. Define
\[
  u_k(x)=x^k-\sum_{j<k}\ip{x^k}{\psi_j}_P\psi_j(x).
\]
By construction, $u_k$ is orthogonal to the span of the previous basis elements.
It is a polynomial of degree exactly $k$, since no $\psi_j$ with $j<k$ contains
an $x^k$ term. Assumption A3 implies that $u_k$ is nonzero in $L^2(P)$; otherwise
$x^k$ would lie in $\Span\{1,\ldots,x^{k-1}\}$ in $L^2(P)$. Normalize
$u_k$ and, if necessary, multiply by $-1$ to make the leading coefficient
positive.

The recursion uses inner products of polynomials of degree at most $k$, hence
moments through order $2k$. For uniqueness, suppose another family satisfies the
same properties. At step $k$, both $k$th polynomials lie in the one-dimensional
orthogonal complement of $\Pi_{k-1}$ inside $\Pi_k$, so they differ only by sign.
The positive-leading orientation fixes that sign.
\end{proof}

\begin{remark}[Classical status]
\Cref{lem:vwk-basis} is classical: existence and uniqueness of the orthonormal
family generated by a positive-definite moment sequence is standard
\citep{Szego1939,Chihara1978}, and building it numerically from the moments of an
arbitrary law is arbitrary polynomial chaos \citep{OladyshkinNowak2012}. We use it
as the computational building block of the Kunchenko construction and claim no novelty for it.
\end{remark}

\begin{proposition}[Kunchenko normal system and matched-basis diagonalization]
\label{prop:normalsystem}
Let $V^\star=\sum_{j=0}^{s}K_j x^j$ be the $L^2(P)$ projection of an output $Y$
onto $\Pi_s$. Then the monomial coordinate vector $K=(K_0,\dots,K_s)^\top$ solves
the Kunchenko normal system
\[
  H_s K = B,\qquad (H_s)_{ij}=m_{i+j},\quad B_i=\E[Y X^i],
\]
in which the Hankel moment (Gram) matrix $H_s$ is exactly the operator $F$ of
$FK=B$ \citep{Kunchenko2002,Kunchenko2006,Zabolotnii2026Survey}. The matched VWK
change of basis $T_s$, whose $k$th row holds the monomial coefficients of
$\psi_k$, is the inverse-Cholesky/Gram--Schmidt factor of $H_s$, namely
$T_s H_s T_s^\top = I$. In the matched basis the Gram operator is the identity, so
the projection is recovered coordinate-wise, $a_k=\E[Y\psi_k(X)]$, with
$a=T_s^{-\top}K$. The variance-matched Gaussian/Wiener estimator performs the same
diagonal solve but with the wrong moment matrix---the Gaussian Hankel $H_s^{N}$ in
place of $H_s$---which is the mechanism behind \cref{prop:penalty}.
\end{proposition}

\begin{proof}
Orthogonality of the residual to each monomial gives
$\E[(Y-\sum_j K_j X^j)X^i]=0$, that is $\sum_j m_{i+j}K_j=\E[YX^i]$, i.e.\ $H_sK=B$;
positive definiteness of $H_s$ under \cref{ass:moment} makes $K$ unique. Since
$\psi_k=\sum_j (T_s)_{kj}x^j$, one has
$(T_sH_sT_s^\top)_{ij}=\ip{\psi_i}{\psi_j}_P=\delta_{ij}$, and $T_s$ is lower
triangular by \cref{lem:vwk-basis}; hence $T_sH_s=T_s^{-\top}$. Then
$a_k=\ip{Y}{\psi_k}_P=\sum_j (T_s)_{kj}B_j=(T_sB)_k$, so
$a=T_sB=T_sH_sK=T_s^{-\top}K$. The Wiener basis $g_k$ is orthonormal for
$N(0,\sigma^2)$, with Gram matrix $H_s^{N}\neq H_s$ under $P$; its diagonal
coefficients $\E[Yg_k(X)]$ thus reconstruct $(H_s^{N})^{-1}B$, solving the normal
system with $H_s^{N}$, and the discrepancy $H_s-H_s^{N}$ is governed at order two
by the skew coefficient $\delta=\mu_3/\sigma^2$. The penalty of \cref{prop:penalty}
is stated for the self-normalized diagonal variant of this Gaussian estimator---the
population limit of the sample cross-correlation coefficients---which coincides with
the fixed-norm reconstruction here at orders $0$ and $1$ and is the conservative
choice at order $2$ (see the remark after \cref{prop:penalty}).
\end{proof}

\section{Finite-memory Volterra expansion}
\label{sec:volterra}

For a multi-index $\alpha=(\alpha_1,\ldots,\alpha_d)$ define the tensor basis
function
\[
  \Psi_\alpha(x_1,\ldots,x_d)=\prod_{r=1}^d \psi_{\alpha_r}(x_r),
\]
and let
\[
  \mathcal{I}_{d,s}=\{\alpha\in\mathbb{N}^d:|\alpha|\le s\}.
\]

\begin{lemma}[Tensor VWK basis]
\label{lem:tensor}
Under \cref{ass:moment}, the family
$\{\Psi_\alpha:\alpha\in\mathcal{I}_{d,s}\}$ is an orthonormal basis of
$\Pi_{d,s}$ in $L^2(P^d)$.
\end{lemma}

\begin{proof}
For $\alpha,\beta\in\mathcal{I}_{d,s}$, independence under $P^d$ gives
\[
  \ip{\Psi_\alpha}{\Psi_\beta}_{P^d}
  =\prod_{r=1}^d\ip{\psi_{\alpha_r}}{\psi_{\beta_r}}_P
  =\prod_{r=1}^d\delta_{\alpha_r\beta_r}.
\]
Thus the tensor functions are orthonormal. They span $\Pi_{d,s}$ because the
one-dimensional change of basis between monomials and $\psi_k$ is triangular with
nonzero diagonal entries, and the tensor-product change of basis restricted to
total degree $\le s$ remains block triangular with nonzero diagonal blocks.
\end{proof}

A finite-memory Volterra polynomial in this paper is
\[
  V_h(x_1,\ldots,x_d)
  =\sum_{\beta\in\mathcal{I}_{d,s}}h_\beta x^\beta .
\]
Its VWK expansion is
\[
  V_a(x_1,\ldots,x_d)
  =\sum_{\alpha\in\mathcal{I}_{d,s}}a_\alpha \Psi_\alpha(x_1,\ldots,x_d).
\]

\begin{theorem}[Finite-memory coefficient mapping]
\label{thm:mapping}
The map between monomial Volterra coefficients $h_\beta$ and VWK coefficients
$a_\alpha$ is linear and bijective on $\Pi_{d,s}$.
\end{theorem}

\begin{proof}
Let $T_s$ be the one-dimensional matrix whose $k$th row contains the monomial
coefficients of $\psi_k$. By \cref{lem:vwk-basis}, $T_s$ is triangular with
nonzero diagonal entries and hence invertible. For memory length $d$, the
corresponding tensor-product basis change is the restriction of $T_s^{\otimes d}$
to the total-degree subspace. Ordering terms by total degree makes this
restriction block triangular, again with nonzero diagonal entries inherited from
the one-dimensional leading coefficients. Therefore the restricted change of
basis is invertible.
\end{proof}

\begin{theorem}[Finite-memory VWK projection]
\label{thm:projection}
Every $V\in\Pi_{d,s}$ has a unique expansion
\[
  V(X)=\sum_{\alpha\in\mathcal{I}_{d,s}} a_\alpha\Psi_\alpha(X),
  \quad X\sim P^d.
\]
If $Y=V(X)+\varepsilon$ and
$\E[\varepsilon\Psi_\alpha(X)]=0$ for all $\alpha\in\mathcal{I}_{d,s}$, then
\[
  a_\alpha=\E[Y\Psi_\alpha(X)].
\]
Without this orthogonality condition on $\varepsilon$, the same formula gives
the $L^2(P^d)$ projection coefficient of $Y$ onto $\Pi_{d,s}$.
\end{theorem}

\begin{proof}
The first claim follows from \cref{lem:tensor}. Orthonormal expansion gives
$a_\alpha=\ip{V}{\Psi_\alpha}_{P^d}$. If
$Y=V+\varepsilon$ and the residual is orthogonal to every basis function, then
$\E[Y\Psi_\alpha]=a_\alpha$. Otherwise, orthogonal projection in a finite
Hilbert space gives the final statement.
\end{proof}

\begin{proposition}[Asymptotic normality of the diagonal estimator]
\label{prop:clt}
Let $(X_i,Y_i)$, $i=1,\dots,n$, be iid with $X_i\sim P^d$, and assume
$\E[Y^2\Psi_\alpha(X)^2]<\infty$ and $\E[\Psi_\alpha(X)^4]<\infty$. Fix
$\alpha\in\mathcal{I}_{d,s}$ and let $a_\alpha=\E[Y\Psi_\alpha(X)]$ be the
population coefficient of \cref{thm:projection}. The diagonal
cross-correlation estimator
\[
  \widehat a_\alpha
  =\frac{\sum_{i=1}^n Y_i\Psi_\alpha(X_i)}{\sum_{i=1}^n \Psi_\alpha(X_i)^2}
\]
satisfies
\[
  \sqrt{n}\,(\widehat a_\alpha-a_\alpha)\ \xrightarrow{d}\ N(0,V_\alpha),
  \qquad
  V_\alpha=\Var\!\big(\Psi_\alpha(X)\,(Y-a_\alpha\Psi_\alpha(X))\big),
\]
where $\E[\Psi_\alpha(X)^2]=1$ by the orthonormality of \cref{lem:tensor}. If,
moreover, $Y=\sum_{\beta\in\mathcal{I}_{d,s}}a_\beta\Psi_\beta(X)+\varepsilon$
with $\E[\varepsilon\mid X]=0$ and $\Var(\varepsilon\mid X)=\sigma_\varepsilon^2$,
then
\[
  V_\alpha=\sigma_\varepsilon^2
  +\sum_{\beta,\gamma\neq\alpha}
     a_\beta a_\gamma\,\E\!\big[\Psi_\alpha^2\,\Psi_\beta\,\Psi_\gamma\big].
\]
\end{proposition}

\begin{proof}
Write $N_n=n^{-1}\sum_i Y_i\Psi_\alpha(X_i)$ and
$D_n=n^{-1}\sum_i\Psi_\alpha(X_i)^2$, so $\widehat a_\alpha=g(N_n,D_n)$ with
$g(u,v)=u/v$. The population values of the two sample means are
\[
  \bigl(\E[Y\Psi_\alpha],\E[\Psi_\alpha^2]\bigr)=(a_\alpha,1),
\]
the second strictly positive. The moment hypotheses make the covariance of the iid vectors
$(Y_i\Psi_\alpha(X_i),\Psi_\alpha(X_i)^2)$ finite, so the bivariate central limit
theorem and the delta method apply. The gradient of $g$ at $(a_\alpha,1)$ is
$(1,-a_\alpha)$, giving influence function
\[
  \phi
  =(Y\Psi_\alpha-a_\alpha)-a_\alpha(\Psi_\alpha^2-1)
  =\Psi_\alpha(Y-a_\alpha\Psi_\alpha).
\]
This is the score of the exactly identified moment condition
$\E[\Psi_\alpha(Y-a_\alpha\Psi_\alpha)]=0$ \citep{Hansen1982}; since
$\E[\phi]=0$, we obtain $\sqrt n(\widehat a_\alpha-a_\alpha)\xrightarrow{d}N(0,V_\alpha)$
with $V_\alpha=\E[\phi^2]=\Var(\Psi_\alpha(Y-a_\alpha\Psi_\alpha))$. Under the
orthogonal-residual model,
$Y-a_\alpha\Psi_\alpha=\sum_{\beta\neq\alpha}a_\beta\Psi_\beta+\varepsilon$;
squaring, taking expectations, using $\E[\varepsilon\mid X]=0$ to annihilate the
$\varepsilon$-cross terms and $\E[\Psi_\alpha^2]=1$ for the noise term yields the
displayed decomposition.
\end{proof}

\begin{remark}[Reported standard errors]
The coverage columns of \cref{tab:h4} use the plug-in delta-method estimator of
$V_\alpha$, namely the sample second moment of the fitted scores
$\Psi_\alpha(X_i)(Y_i-\widehat a_\alpha\Psi_\alpha(X_i))$ divided by
$\big(n^{-1}\sum_i\Psi_\alpha(X_i)^2\big)^2$. Its cross-feature contribution is
carried by the triple products $\E[\Psi_\alpha^2\Psi_\beta\Psi_\gamma]$, which
vanish for symmetric $P$ but are nonzero under skew; a nominal interval that
ignores them is too short, which accounts for the sub-nominal coverage of the
centered-exponential rows. The mild undercoverage of the symmetric-contaminated
row (\cref{tab:h4}, $0.885$) has a different source: its triple products vanish,
and the shortfall instead reflects the slow convergence of the plug-in variance
estimate under the heavy contaminated tails, where the fourth-to-$4s$ input
moments that enter $\widehat V_\alpha$ are estimated with large sampling noise.
With the empirical (random) basis at memory $d=1$, $\widehat a_\alpha$ coincides
numerically with ordinary least squares in the same span.
\end{remark}

\begin{figure}[t]
\centering
\resizebox{\textwidth}{!}{%
\begin{tikzpicture}[
  node distance=0.95cm and 1.05cm,
  box/.style={draw, rounded corners=2pt, align=center, minimum width=2.6cm,
              minimum height=0.9cm, font=\small, fill=gray!7},
  arr/.style={-{Latex[length=2mm]}, thick}
]
\node[box] (law) {input law\\$P$};
\node[box, right=of law] (mom) {moments\\$m_0,\ldots,m_{2s}$};
\node[box, right=of mom] (gs) {oriented\\Gram--Schmidt};
\node[box, right=of gs] (basis) {basis\\$\psi_0,\ldots,\psi_s$};
\node[box, below=of basis] (tensor) {tensor features\\$\Psi_\alpha=\prod_r\psi_{\alpha_r}$};
\node[box, left=of tensor] (proj) {Volterra projection\\$a_\alpha=\E[Y\Psi_\alpha]$};
\node[box, left=of proj] (map) {linear map\\$a\leftrightarrow h$};
\draw[arr] (law) -- (mom);
\draw[arr] (mom) -- (gs);
\draw[arr] (gs) -- (basis);
\draw[arr] (basis) -- (tensor);
\draw[arr] (tensor) -- (proj);
\draw[arr] (proj) -- (map);
\end{tikzpicture}}
\caption{Finite-memory VWK pipeline: the input law determines the orthogonal
polynomial basis, tensor products produce Volterra features, and projection
coefficients are linearly related to monomial Volterra coefficients.}
\label{fig:pipeline}
\end{figure}
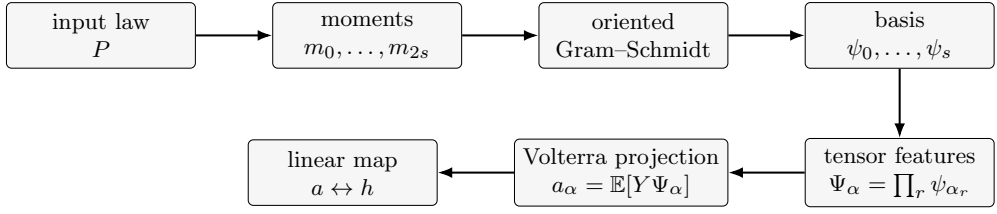

\section{Classical reductions}
\label{sec:reductions}

\subsection{Gaussian law and Wiener--Hermite}

\begin{theorem}[Gaussian reduction]
\label{thm:gaussian}
If $P=N(0,\sigma^2)$ with $\sigma>0$, the oriented VWK basis of
\cref{lem:vwk-basis} is
\[
  \psi_k(x)=\frac{\He_k(x/\sigma)}{\sqrt{k!}},
\]
where $\He_k$ is the probabilists' Hermite polynomial.
\end{theorem}

\begin{proof}
Let $Z=X/\sigma\sim N(0,1)$. Classical Hermite orthogonality gives
\[
  \E[\He_i(Z)\He_j(Z)]=i!\,\delta_{ij}.
\]
Therefore $\He_k(x/\sigma)/\sqrt{k!}$ is orthonormal in
$L^2(N(0,\sigma^2))$. It has degree $k$ and positive leading coefficient
$\sigma^{-k}/\sqrt{k!}$. By the uniqueness part of \cref{lem:vwk-basis}, it is
the VWK basis polynomial.
\end{proof}

Thus Wiener--Hermite \citep{Wiener1958,CameronMartin1947} is recovered as the
Gaussian instance of the distribution-indexed (Wiener--Askey) construction
\citep{Xiu2002}, not displaced; this is a consistency check rather than a
contribution.

\subsection{Askey reductions}

\begin{theorem}[Askey rows as closed VWK instances]
\label{thm:askey}
For each law in \cref{tab:askey}, the oriented VWK basis equals the normalized
classical orthogonal-polynomial family for that law, with sign fixed by positive
leading coefficient.
\end{theorem}

\begin{proof}
Each row supplies a known polynomial family $q_k$ with degree $k$ and
$\ip{q_i}{q_j}_P=n_k\delta_{ij}$, $n_k>0$. Normalizing $q_k$ by $n_k^{-1/2}$ and
fixing the leading sign gives a family satisfying all defining properties of
\cref{lem:vwk-basis}. Uniqueness identifies it with the VWK basis.
\end{proof}

The distribution-to-family correspondence in \cref{tab:askey} is the Wiener--Askey
scheme \citep{Xiu2002,Koekoek2010}, classified for Sheffer systems by
\citet{Schoutens2000}; we recover it as a consistency check, not as a new identity.

\begin{table}[t]
\centering
\caption{Classical distributions and closed VWK instances.}
\label{tab:askey}
\begin{tabular}{@{}lll@{}}
\toprule
Input law $P$ & VWK basis family & Support note \\
\midrule
Gaussian & Hermite & $\R$ \\
Binomial & Krawtchouk & $\{0,\ldots,N\}$, so $s\le N$ \\
Poisson & Charlier & $\mathbb{N}_0$ \\
Negative binomial & Meixner & $\mathbb{N}_0$ \\
Gamma / exponential & Laguerre & $\R_+$ \\
Beta / uniform & Jacobi / Legendre & bounded interval \\
\bottomrule
\end{tabular}
\end{table}

\begin{proposition}[Binomial reduction and machine-checked Krawtchouk orthogonality]
\label{prop:krawtchouk}
Let $P=\mathrm{Binomial}(N,p)$ with $0<p<1$, centered to mean zero. The matched
VWK basis $\{\psi_0,\psi_1,\psi_2\}$ of \cref{lem:vwk-basis} coincides with the
monic-normalized Krawtchouk family $\{K_0,K_1,K_2\}$ \citep{Koekoek2010}, and the
skew coefficient is
\[
  \delta=\frac{\mu_3}{\sigma^2}=1-2p,
\]
which is independent of $N$ and vanishes exactly at $p=\tfrac12$. The orthogonality
$\ip{K_i}{K_j}_P=0$ for $i\neq j$ holds for \emph{arbitrary} $N$ and is
machine-checked in Lean~4 (theorems \texttt{orth12} and \texttt{central3\_N}).
\end{proposition}

Closing $K_1\perp K_2$ at general $N$ reduces to the third Bernstein factorial moment
$\sum_{\nu}\nu(\nu-1)(\nu-2)\,b_{N,\nu}(x)=N(N-1)(N-2)\,x^3$, which is absent from
Mathlib and is supplied as a candidate upstream lemma. The statement is therefore a
machine-checked instance of the matched-basis $=$ classical-family recovery of
\cref{thm:askey}, with $\delta=1-2p$ the discrete analogue of the skew coefficient
that governs the order-$2$ penalty of \cref{prop:penalty}.

\section{Misspecification penalty}
\label{sec:penalty}

The construction so far is a recovery of arbitrary polynomial chaos. The result of
this section is the paper's analytic contribution: a closed-form account of what is
lost by projecting in the variance-matched Gaussian/Wiener basis when the input law
is not Gaussian, quantified at order two by the skew coefficient
$\delta=\mu_3/\sigma^2$.

\begin{proposition}[Order-2 misspecification penalty]\label{prop:penalty}
Let $X\sim P$ have $\E[X]=0$, $\Var(X)=\sigma^2>0$ and finite moments $\mu_3,\mu_4$.
Let $\{\psi_0,\psi_1,\psi_2\}$ be the matched $\VWK$ basis, orthonormal in $L^2(P)$, and let
$g_k(x)=\He_k(x/\sigma)/\sqrt{k!}$ be the variance-matched Gaussian basis. Set
\begin{equation*}
  \lambda:=\frac{\mu_3}{\sigma},\qquad
  \rho^2:=\mu_4-\sigma^4-\frac{\mu_3^2}{\sigma^2}\ \ (\geq 0),\qquad
  \delta:=\frac{\mu_3}{\sigma^2}.
\end{equation*}
For a signal $f=\beta_0\psi_0+\beta_1\psi_1+\beta_2\psi_2$, the matched diagonal projection
$a_k=\ip{f}{\psi_k}_P$ recovers $f$ exactly, whereas the variance-matched Gaussian
cross-correlation estimator in its self-normalized diagonal form
\begin{equation*}
  \hat f_W=\sum_{k=0}^{2}\frac{\ip{f}{g_k}_P}{\ip{g_k}{g_k}_P}\,g_k
\end{equation*}
incurs the excess $L^2(P)$ risk
\begin{equation}\label{eq:penalty}
  \bigl\lVert \hat f_W-f \bigr\rVert_P^2
  =(\gamma_2\lambda)^2+(\gamma_2\rho-\beta_2)^2,
  \qquad
  \gamma_2:=\frac{\beta_1\lambda+\beta_2\rho}{\lambda^2+\rho^2}.
\end{equation}
The penalty vanishes for every nondegenerate signal if and only if $\mu_3=0$.
\end{proposition}

\begin{proof}
Since $\{\psi_k\}$ is orthonormal in $L^2(P)$, the matched coordinates
$a_k=\ip{f}{\psi_k}_P=\beta_k$ reproduce $f$, so $V:=\lVert f-\sum_k a_k\psi_k\rVert_P^2=0$.

For the Gaussian basis, $g_0=\He_0(x/\sigma)=1=\psi_0$ and
$g_1=\He_1(x/\sigma)=x/\sigma=\psi_1$, while
$g_2=\He_2(x/\sigma)/\sqrt{2}=(x^2-\sigma^2)/(\sigma^2\sqrt{2})$.
Writing $x^2=\pi_2+\delta x+\sigma^2$ with $\pi_2=\rho\,\psi_2$ and $x=\sigma\psi_1$, and using
$\delta\sigma=\lambda$, gives
\begin{equation*}
  g_2=\frac{\rho\,\psi_2+\lambda\,\psi_1}{\sigma^2\sqrt{2}},
  \qquad
  \ip{g_2}{g_2}_P=\frac{\lambda^2+\rho^2}{2\sigma^4},
\end{equation*}
the last equality by orthonormality of $\psi_1,\psi_2$. The diagonal coefficients
$b_k=\ip{f}{g_k}_P/\ip{g_k}{g_k}_P$ therefore satisfy $b_0=\beta_0$, $b_1=\beta_1$, and
\begin{equation*}
  b_2=\frac{\ip{f}{g_2}_P}{\ip{g_2}{g_2}_P}
     =\frac{(\beta_1\lambda+\beta_2\rho)/(\sigma^2\sqrt{2})}{(\lambda^2+\rho^2)/(2\sigma^4)}
     =\sqrt{2}\,\sigma^2\,\gamma_2,
\end{equation*}
so that $b_2 g_2=\gamma_2(\rho\,\psi_2+\lambda\,\psi_1)$. Collecting terms,
\begin{equation*}
  \hat f_W-f=(\gamma_2\lambda)\,\psi_1+(\gamma_2\rho-\beta_2)\,\psi_2,
\end{equation*}
and orthonormality of $\psi_1,\psi_2$ yields \eqref{eq:penalty}. If $\mu_3=0$ then $\lambda=0$,
whence $\gamma_2\rho=\beta_2$ and both bias terms vanish; conversely, for $\mu_3\neq0$ (so
$\lambda\neq0$) a vanishing right-hand side of \eqref{eq:penalty} forces $\gamma_2=0$ and
$\beta_2=0$, hence $\beta_1=0$, i.e.\ the degenerate signal. This proves the stated equivalence.
\end{proof}

\begin{remark}\label{rem:penalty}
Both bias terms in \eqref{eq:penalty} carry the factor $\lambda=\mu_3/\sigma$, so the penalty
vanishes identically for symmetric inputs and otherwise scales with the skew coefficient
$\delta=\mu_3/\sigma^2$. Because the reported penalty depends on which Gaussian estimator is
meant, we fix it here. First, $\hat f_W$ is the \emph{self-normalized} diagonal projection
$b_k=\ip{f}{g_k}_P/\ip{g_k}{g_k}_P$: it is the population limit of the practical sample
cross-correlation estimator, which normalizes each Gaussian coordinate by its realized second
moment $n^{-1}\sum_i g_k(X_i)^2$. Second, it coincides with the classical fixed-norm Lee--Schetzen
construction \citep{LeeSchetzen1965} of \cref{prop:normalsystem} at orders $0$ and $1$---there
$\ip{g_k}{g_k}_P=1$---and differs only at order $2$, where the classical estimator divides instead
by the fixed Gaussian norm $\ip{g_2}{g_2}_{N(0,\sigma^2)}=1$. Third, among the two the
self-normalized form is the one most favorable to the Gaussian basis, because dividing by the true
$P$-norm $\ip{g_2}{g_2}_P$---which the classical procedure does not know---is the diagonal
coefficient that minimizes the order-two risk; hence \eqref{eq:penalty} is a \emph{lower bound} on
the excess risk of the fixed-norm Lee--Schetzen estimator. Even under this conservative accounting,
matching the basis to $P$ is what restores validity off the Gaussian input. For
$X\sim\mathrm{Binomial}(N,p)$ one has
$\delta=1-2p$, independent of $N$ (cf.\ \cref{prop:krawtchouk}). The closed form \eqref{eq:penalty}
agrees with an exact moment computation and with Monte Carlo to machine precision:
for the centered-exponential row
$\sigma^2=1$, $\mu_3=2$, $\mu_4=9$ one obtains $\lambda=\rho=2$, $\gamma_2=0.35$, and excess risk
$0.7^2+0.1^2=0.500000$, reproducing the ratio $W/V=32.36$ of \cref{tab:h4}.
\end{remark}

\begin{remark}[General order]\label{rem:general-s}
The mechanism is not special to $s=2$. For any order $s$ the mismatched estimator performs the
diagonal solve with the Gaussian Hankel $H_s^{N}$ in place of $H_s$ (\cref{prop:normalsystem}), so
in the matched coordinates its output is $\hat f_W=D_s f$ for a fixed bounded operator $D_s$ that
equals the identity if and only if the moments of $P$ agree with the Gaussian ones through order
$2s$; the excess risk is then $\lVert (I-D_s)f\rVert_P^2$. We give the closed form only at $s=2$
because $D_2$ depends on moments through $\mu_4$ and, being the first order at which the Gaussian
and matched second-degree polynomials part company, it isolates the leading skew-activated term
$\lambda=\mu_3/\sigma$. A general-$s$ closed form requires the full moment vector
$(m_1,\dots,m_{2s})$ and couples every degree $\le s$ through $D_s$; symmetry of $P$ through the
relevant order removes the coupling exactly, consistent with the $\mu_3=0$ equivalence of
\cref{prop:penalty}.
\end{remark}

\section{Constructive estimator and reproducibility}
\label{sec:estimator}

\paragraph{Algorithm~1 (Matched \VWK{} Volterra projection).}
\emph{Input:} moments $m_0,\dots,m_{2s}$ of the input law (or a sample
$u_1,\dots,u_n$, from which empirical moments $\hat m_k=\E_n[u^k]$ are formed),
memory order $d$, polynomial degree $s$, and regression pairs
$(\mathbf{x}_t,Y_t)$ with lag vector
$\mathbf{x}_t=(u_t,u_{t-1},\dots,u_{t-d+1})$.
\begin{enumerate}[leftmargin=*]
  \item Assemble the Hankel matrix $H_s=(m_{i+j})_{0\le i,j\le s}$ and compute its
  triangular factor $T_s$ satisfying $T_s H_s T_s^\top=I$; equivalently, apply
  oriented Gram--Schmidt to $1,x,\dots,x^s$ in $\ip{\cdot}{\cdot}_P$ to obtain the
  matched basis $\psi_0,\dots,\psi_s$ with positive leading coefficients
  (\cref{tab:notation}).
  \item For each $\alpha\in \mathcal{I}_{d,s}$ form the tensor feature
  $\Psi_\alpha(\mathbf{x}_t)=\prod_{r=1}^{d}\psi_{\alpha_r}(x_{t,r})$ across the
  rows $\mathbf{x}_t$ of the lag matrix.
  \item Project diagonally: $a_\alpha=\E_n[\,Y\,\Psi_\alpha\,]$ for each
  $\alpha\in \mathcal{I}_{d,s}$, where orthonormality reduces every coordinate to a
  single empirical correlation.
  \item \emph{(Optional ridge.)} For ridge level $\kappa\ge 0$ replace $a_\alpha$ by
  $a_\alpha/(1+\kappa)$; since the matched Gram is $I$, the penalty acts
  coordinatewise.
  \item Map between bases through $T_s$:
  \[
    h_\beta=\sum_{\alpha\in \mathcal{I}_{d,s}} a_\alpha\prod_{r=1}^{d}(T_s)_{\alpha_r,\beta_r},
    \qquad
    a_\alpha=\sum_{\beta\in \mathcal{I}_{d,s}} h_\beta\prod_{r=1}^{d}(T_s^{-1})_{\beta_r,\alpha_r}.
  \]
\end{enumerate}
\emph{Output:} \VWK{} coordinates $\{a_\alpha\}_{\alpha\in \mathcal{I}_{d,s}}$ and the
equivalent monomial Volterra kernel $\{h_\beta\}_{\beta\in \mathcal{I}_{d,s}}$, encoding the
same map
$V=\sum_{\beta} h_\beta\,x^\beta=\sum_{\alpha} a_\alpha\Psi_\alpha$.

\paragraph{Complexity.}
The feature count is $\lvert \mathcal{I}_{d,s}\rvert=\binom{d+s}{s}$. Building the matched
univariate basis is a single $O(s^3)$ triangular factorization of $H_s$, reused
across all $d$ taps; evaluating the $\lvert \mathcal{I}_{d,s}\rvert$ tensor features and
forming the diagonal projection costs $O\!\left(n\,\lvert \mathcal{I}_{d,s}\rvert\right)$,
because orthonormality reduces each coordinate to one correlation. Ordinary
monomial least squares over the same index set instead assembles and inverts a
dense $\lvert \mathcal{I}_{d,s}\rvert\times\lvert \mathcal{I}_{d,s}\rvert$ normal system at cost
$O\!\left(n\,\lvert \mathcal{I}_{d,s}\rvert^{2}+\lvert \mathcal{I}_{d,s}\rvert^{3}\right)$, and inherits
the conditioning of $H_s$, whose condition number grows by roughly a decade per
degree (\cref{subsec:conditioning}); the matched design has Gram $I$ by
construction, with condition number $1$.

\section{Experiments}
\label{sec:experiments}

The experiments test the projection claim of \cref{prop:penalty}, not universal
prediction dominance. Throughout, the reported VWK and Wiener risks are both
\emph{diagonal} (coordinate-wise) projections over the \emph{same}
degree-${\le}\,s$ polynomial span: each coefficient is the cross-correlation
$\ip{Y}{\phi_k}_P/\ip{\phi_k}{\phi_k}_P$ formed independently, which is the
Lee--Schetzen functional read off in the chosen basis $\{\phi_k\}$
\citep{LeeSchetzen1965}. Accordingly we define $W/V$ as the population ratio of
the variance-matched Gaussian diagonal-projection risk to the matched-VWK
diagonal-projection risk,
\[
  W/V=\frac{\text{diagonal-projection risk in the variance-matched Gaussian basis } g_k}
           {\text{diagonal-projection risk in the matched VWK basis } \psi_k},
\]
with finite-memory variants $W/\mathrm{oracle}$ and $W/\mathrm{empirical}$.
Because $g_k$ and $\psi_k$ span the identical subspace, this ratio isolates the
cross-correlation estimator under a mismatched inner product
(\cref{prop:penalty}), not the expressive power of the span. Direct monomial
least squares and Huber \citep{Huber1964} fits, which fit the same span without
the diagonal restriction, are reported alongside, and \cref{subsec:deconfound}
uses full least squares to separate the estimator effect from the span explicitly.

\subsection{Data-generating specification}
\label{subsec:dgp}

Every synthetic table below is generated by the specification in
\cref{tab:dgp}, and the driver scripts, frozen output files, and Lean sources are
released as supplementary material so that each table reproduces exactly from the
seed. The true signal is fixed across regimes as
$f=0.2\,\psi_0+0.8\,\psi_1+0.6\,\psi_2$ in the matched order-two VWK basis
(so $(\beta_1,\beta_2)=(0.8,0.6)$), observed under additive Gaussian noise of
standard deviation $0.25$. The four input laws are centered and, except for the
contaminated mixture, scaled to unit variance; they span the sign of the skew
coefficient $\delta=\mu_3/\sigma^2$ that governs \cref{prop:penalty}.

\begin{table}[t]
\centering
\caption{Data-generating specification for the synthetic experiments. All input
laws are zero-mean. The centered-exponential law is the only asymmetric regime
($\delta=2$); the Gaussian, uniform, and contaminated laws are symmetric
($\delta=0$), so \cref{prop:penalty} predicts no diagonal penalty for them.}
\label{tab:dgp}
\small
\begin{tabular}{@{}llrrr@{}}
\toprule
Regime & Input law $P$ & $\sigma^2$ & $\mu_3$ & $\delta=\mu_3/\sigma^2$ \\
\midrule
Gaussian control        & $\mathcal N(0,1)$                                  & $1.000$ & $0$ & $0$ \\
Centered exponential    & $\mathrm{Exp}(1)-1$                                & $1.000$ & $2$ & $2$ \\
Uniform platykurtic     & $\mathcal U(-\sqrt3,\sqrt3)$                       & $1.000$ & $0$ & $0$ \\
Symmetric contaminated  & $0.9\,\mathcal N(0,0.25)+0.1\,\mathcal N(0,4.0)$   & $0.625$ & $0$ & $0$ \\
\bottomrule
\end{tabular}
\end{table}

Population moments $\mu_k$ for each law are computed in closed form and passed to
the matched Gram--Schmidt construction; the one-lag grid uses $n=2000$ samples,
$200$ repetitions, and seed $20260608$, and the finite-memory grid uses the
parameters stated in \cref{subsec:finite-memory}. Coefficient confidence
intervals use the ratio-estimator (delta-method) standard error for the diagonal
projection $\hat\beta_k=\ip{Y}{\psi_k}_P/\ip{\psi_k}{\psi_k}_P$, which accounts
for the observation-noise contribution that a residual-only conditional standard
error would miss; the reported coverage is the empirical rate at which the
$95\%$ normal interval $\hat\beta_k\pm1.96\,\mathrm{se}$ contains the true
$\beta_k$ across repetitions.

\subsection{Numerical reductions}
\label{subsec:numerical}

The discrete Askey verification covers Krawtchouk, Charlier, and Meixner
polynomials for degrees $0,\ldots,5$. The worst discrepancy between the VWK
Gram--Schmidt basis and the classical normalized basis is $3.65\times10^{-12}$.
The continuous verification covers Hermite, Laguerre, and Legendre polynomials,
with worst discrepancy $3.93\times10^{-14}$ in the stored production report
(the rerun in the current environment reports $3.95\times10^{-14}$, the same
machine-precision scale). These checks support \cref{thm:askey}; they are
numerical verification, not formal proof of the theorem.

The Lean formalization covers selected instances: closed forms for
$\psi_0,\psi_1,\psi_2$ in the Gaussian/Hermite row, Gaussian moments through
order four, orthonormality for $\{\psi_0,\psi_1,\psi_2\}$, and the arbitrary-$N$
Binomial$\to$Krawtchouk orthogonality of \cref{prop:krawtchouk}, which is
obtained by lifting Mathlib's Bernstein-polynomial moment identities evaluated
at $X=p$ and rests on the third Bernstein factorial-moment lemma stated there.
This is formal support for selected instances, not full formal verification of
the construction.

\subsection{One-lag synthetic grid}
\label{subsec:h4}

The production one-lag grid uses the specification of \cref{tab:dgp}
($200$ repetitions, $n=2000$, seed $20260608$). The Gaussian control correctly
remains neutral. The centered exponential regime shows a large
misspecified-Wiener penalty, whereas the uniform and symmetric-contaminated
controls do not, exactly as \cref{prop:penalty} predicts from the sign of
$\delta$.

\begin{table}[t]
\centering
\caption{One-lag synthetic production grid. Ratio $W/V>1$ favors the matched VWK
projection over the misspecified Gaussian/Wiener projection.}
\label{tab:h4}
\scriptsize
\setlength{\tabcolsep}{2pt}
\begin{tabular}{@{}lrrrrrr@{}}
\toprule
Regime & VWK MSE & Wiener MSE & Huber MSE & $W/V$ & $\beta_2$ cov. & all-$\beta$ cov. \\
\midrule
Gaussian control & 0.003248 & 0.003248 & 0.000098 & 1.000 & 0.950 & 0.950 \\
Centered exponential skew & 0.015948 & 0.516133 & 0.000094 & 32.362 & 0.875 & 0.915 \\
Uniform platykurtic & 0.001622 & 0.001622 & 0.000104 & 1.000 & 0.925 & 0.938 \\
Symmetric contaminated & 0.016033 & 0.016033 & 0.000099 & 1.000 & 0.885 & 0.912 \\
\bottomrule
\end{tabular}
\end{table}

The result is strong evidence for an asymmetric finite-moment projection gap.
It is not evidence for a symmetric non-Gaussian advantage, and it is not a claim
that VWK beats robust or likelihood baselines as a general predictor.

\subsection{Estimator de-confounding}
\label{subsec:deconfound}

The $W/V$ ratio of \cref{prop:penalty} compares two diagonal estimators, so
before reading it as a property of the matched basis one must exclude the
possibility that the misspecified Gaussian projection merely fits a poorer span.
Full least squares over a fixed span is basis-invariant, hence least squares in
the VWK, variance-matched Gaussian/Wiener, and raw monomial bases must coincide.
On the centered-exponential regime of \cref{tab:h4} (same seed, $n=2000$, $200$
repetitions), all three full-least-squares fits return signal MSE $0.000088$,
agreeing to a maximum pairwise difference of $4.5\times10^{-19}$
(\cref{tab:deconfound}), whereas the diagonal projections return the
\cref{tab:h4} values $0.015948$ in the matched VWK basis and $0.516133$ in the
variance-matched Gaussian basis, a diagonal ratio $W/V=32.36$. Full least squares
therefore sits roughly $181\times$ below even the matched diagonal projection,
which establishes that the $W/V$ gap is an estimator-in-a-mismatched-basis effect
rather than a deficiency of the degree-${\le}\,2$ span. At memory $d=1$ the empirical-moment
VWK diagonal projection coincides numerically with ordinary least squares over
the same span, so in the single-lag case the matched basis recovers the
full-projection estimate exactly.

\begin{table}[t]
\centering
\caption{De-confounding the $W/V$ ratio on the centered-exponential regime
($n=2000$, $200$ repetitions, additive noise standard deviation $0.25$). Diagonal
projection is coordinate-wise; full least squares fits the whole span. Values are
signal MSE; the three full-least-squares bases agree to $4.5\times10^{-19}$.}
\label{tab:deconfound}
\small
\begin{tabular}{@{}lr@{}}
\toprule
Estimator over the degree-${\le}\,2$ span & Signal MSE \\
\midrule
Diagonal projection, matched VWK basis $\psi_k$            & 0.015948 \\
Diagonal projection, variance-matched Gaussian basis $g_k$ & 0.516133 \\
Full least squares, VWK / Wiener / monomial basis          & 0.000088 \\
\bottomrule
\end{tabular}
\end{table}

The correct reading follows. Matched VWK does not improve on full least squares
for prediction; its contribution is to restore the validity of the inexpensive
diagonal cross-correlation estimator under an asymmetric input law, so that the
coordinate-wise Lee--Schetzen coefficients are unbiased rather than corrupted by
basis mismatch (\cref{prop:penalty}), and to supply the well-conditioned
coordinate system whose regularization behavior is analyzed in
\cref{subsec:conditioning}.

\subsection{Basis conditioning}
\label{subsec:conditioning}

Full least squares is basis-invariant: it fits the same span and returns the
identical predictor whether the features are expressed in monomial, Wiener, or
matched coordinates (\cref{subsec:deconfound}). The durable, regime-specific
advantage of the matched basis is therefore not approximation power but
\emph{conditioning}. The monomial (power) features carry the moment Gram
operator $H_s=(m_{i+j})_{0\le i,j\le s}$, whose condition number grows
geometrically with the polynomial order $s$; the matched $\VWK$ basis carries,
by construction, the Gram $T_s H_s T_s^{\top}=I$, so its condition number is
unity at every order.

\Cref{tab:conditioning} reports $\mathrm{cond}(H_s)$ for the population power
basis, the population $\VWK$ Gram (unity by construction), and---to answer the
estimation-relevant question directly---the median condition number of the
\emph{empirical} $\VWK$ Gram formed from the realized $n=2000$ design over $40$
repetitions. For the centered-exponential law the power Hankel condition number
climbs from $23$ at $s=2$ to $1.63\times10^{9}$ at $s=6$, an
eight-order-of-magnitude deterioration over four degrees. The empirical $\VWK$
Gram is not itself the identity: it departs from unity and, for the
centered-exponential law, degrades from $1.3$ at $s=2$ to $1.1\times10^{4}$ at
$s=6$, since the fourth-to-$4s$ input moments that stabilize the empirical Gram
converge slowly under heavy tails. The durable advantage is therefore not that
the empirical matched Gram is perfectly conditioned, but that it remains four to
five orders of magnitude better conditioned than the power basis at every order
(at $s=6$, median $1.1\times10^{4}$ against $1.6\times10^{9}$). The gap holds
across regimes: at $s=5$ the power Hankel condition number is
$4.90\times10^{3}$ (Gaussian), $1.39\times10^{3}$ (uniform), and
$3.44\times10^{5}$ (contaminated), against empirical $\VWK$ medians of
$3.2$, $1.2$, and $11.4$ respectively.

\begin{table}[t]
\centering
\caption{Basis conditioning: condition number of the population monomial (power)
Hankel Gram $H_s$, the population matched $\VWK$ Gram (the identity by
construction), and the median condition number of the \emph{empirical} $\VWK$
Gram over $40$ repetitions at $n=2000$. The full order sweep is shown for the
centered-exponential law; the remaining regimes are anchored at $s=5$.}
\label{tab:conditioning}
\begin{tabular}{@{}llrrr@{}}
\toprule
Regime & $s$ & Power $\mathrm{cond}(H_s)$ & Pop.\ \VWK & Emp.\ \VWK (med.) \\
\midrule
Centered exponential & 2 & $23$                & $1.000$ & $1.3$ \\
                     & 3 & $9.90\times10^{2}$  & $1.000$ & $2.6$ \\
                     & 4 & $7.25\times10^{4}$  & $1.000$ & $11.9$ \\
                     & 5 & $8.72\times10^{6}$  & $1.000$ & $4.0\times10^{2}$ \\
                     & 6 & $1.63\times10^{9}$  & $1.000$ & $1.1\times10^{4}$ \\
\midrule
Gaussian             & 5 & $4.90\times10^{3}$  & $1.000$ & $3.2$ \\
Uniform              & 5 & $1.39\times10^{3}$  & $1.000$ & $1.2$ \\
Contaminated         & 5 & $3.44\times10^{5}$  & $1.000$ & $11.4$ \\
\bottomrule
\end{tabular}
\end{table}

The conditioning gap propagates to estimation under regularization, but only
under a \emph{common} penalty scale. \Cref{tab:ridge} contrasts two protocols on $40$ repetitions
($n_{\text{train}}=2800$, $n_{\text{test}}=1200$). Under a fixed shared relative
ridge penalty $\eta=10^{-2}$ (with $\lambda=\eta\,\mathrm{tr}(G)/(s+1)$), the
prediction RMSE in the power basis exceeds that in the $\VWK$ basis by factors of
$7.4$ to $25.5$, because a single penalty scale mis-shrinks the anisotropic power
coordinates. Under per-basis cross-validated $\lambda$, however, the advantage
essentially disappears: the power/$\VWK$ RMSE ratio falls to between $0.58$ and
$1.42$, and on two configurations the power basis is slightly better.
The advantage is therefore narrow: the matched basis provides robustness to
\emph{naive or shared} regularization and a well-scaled coordinate system, not a
lower achievable RMSE once each basis is tuned; a practitioner who tunes $\lambda$
per basis recovers comparable prediction error either way, because the underlying
span---and hence the unpenalized least-squares fit---is identical
(\cref{subsec:deconfound}).

\begin{table}[t]
\centering
\caption{Ridge prediction RMSE ratio (power basis over matched $\VWK$ basis) under
two regularization protocols, $40$ repetitions, $n_{\text{train}}=2800$. Under a
common shared penalty $\eta=10^{-2}$ the matched basis is $7$--$26\times$ better;
under per-basis cross-validated $\lambda$ the ratio collapses toward one, showing
the fixed-$\eta$ advantage is a penalty-scaling artifact rather than a lower
achievable error.}
\label{tab:ridge}
\small
\begin{tabular}{@{}llrr@{}}
\toprule
Regime & $s$ & Fixed $\eta$: power/\VWK & CV-tuned: power/\VWK \\
\midrule
Gaussian             & 3 & $7.37$  & $1.05$ \\
                     & 5 & $25.52$ & $1.00$ \\
Centered exponential & 3 & $12.40$ & $1.14$ \\
                     & 5 & $9.94$  & $0.58$ \\
Uniform              & 3 & $10.65$ & $0.94$ \\
                     & 5 & $15.14$ & $1.06$ \\
Contaminated         & 3 & $14.89$ & $0.95$ \\
                     & 5 & $17.78$ & $1.42$ \\
\bottomrule
\end{tabular}
\end{table}

One caveat bounds the construction's reach: raw-moment Gram--Schmidt is itself
ill-conditioned for large $s$ \citep{Gautschi1982}, so we restrict $s$ to the
reported range. With exact population moments the $\VWK$ Gram is stable to machine
precision (orthogonality error $\le 1.4\times10^{-12}$ up to $s=6$); the
finite-sample degradation of the empirical $\VWK$ Gram reported in
\cref{tab:conditioning} is the separate, sampling-driven effect and is the binding
limit in practice. For higher degree a Stieltjes / three-term recurrence is the
standard remedy for both.

\subsection{Finite-memory empirical moments}
\label{subsec:finite-memory}

The second synthetic gate uses a second-order finite-memory system with
$d=3$, $150$ repetitions, $n=2500$, and the same noise level. The VWK estimator is
run both with oracle population moments and with empirical moments estimated
from the input sample. The centered-exponential gap persists under both modes.

\begin{sidewaystable}[t]
\centering
\caption{Finite-memory empirical-moment production grid with $d=3$ and $s=2$.}
\label{tab:finite-memory}
\scriptsize
\setlength{\tabcolsep}{2pt}
\begin{tabular}{@{}lrrrrrrr@{}}
\toprule
Regime & Oracle VWK & Emp. VWK & Wiener & LS & Huber & $W$/oracle & $W$/emp. \\
\midrule
Gaussian control & 0.013221 & 0.009674 & 0.013221 & 0.000254 & 0.000264 & 1.000 & 1.367 \\
Centered exponential skew & 0.050598 & 0.037156 & 0.680900 & 0.000252 & 0.000262 & 13.457 & 18.325 \\
Uniform platykurtic & 0.008060 & 0.005658 & 0.008060 & 0.000247 & 0.000260 & 1.000 & 1.425 \\
Symmetric contaminated & 0.051937 & 0.032785 & 0.051937 & 0.000245 & 0.000254 & 1.000 & 1.584 \\
\bottomrule
\end{tabular}
\end{sidewaystable}

Two features of the empirical-moment column can mislead. First, the
empirical VWK risk is below the oracle VWK risk in every regime, and the
Gaussian control shows $W/\mathrm{empirical}=1.367$ even though it is neutral
under oracle moments ($W/\mathrm{oracle}=1.000$). Neither is a skew effect: the
empirical basis is orthogonalized against the \emph{realized} design, so at
memory $d=1$ it reduces to ordinary least squares in the same span and enjoys a
generic in-sample fitting advantage over any fixed population basis, the Gaussian
control included. The neutrality claim of \cref{prop:penalty} is therefore an
oracle-moment statement; the empirical column measures a different, favorable
finite-sample artifact and must not be read as the penalty. Second, these rows
are not evidence that empirical moments dominate oracle moments;
the estimators are evaluated in a finite-sample projection experiment, and the
skew-driven separation that the penalty predicts is the centered-exponential
column ($W/\mathrm{oracle}=13.457$), which persists under both moment modes.

\subsection{Sample efficiency}
\label{subsec:sample-efficiency}

The sample-efficiency curve repeats the finite-memory experiment with sample
sizes $600$, $1200$, $2500$, and $5000$, using $80$ repetitions per point. The minimum
centered-exponential ratios are $W/\mathrm{oracle}=4.164$ and
$W/\mathrm{empirical}=6.198$, while the Gaussian control has
$W/\mathrm{oracle}=1.000$ at all sample sizes. The shaded bands are $95\%$
percentile bootstrap confidence intervals over the $80$ replications,
resampling replication indices jointly so the paired Wiener and VWK draws stay
together. The skew bands lie entirely above the neutral ratio of one at every
sample size---the oracle-moment lower endpoints rise from $3.50$ at $n=600$ to
$23.04$ at $n=5000$---so the separation is not a sampling artifact; the
Gaussian-control band is degenerate at one because the variance-matched Wiener
basis coincides with the VWK oracle basis replicate by replicate.

\begin{figure}[t]
\centering
\includegraphics[width=0.82\linewidth]{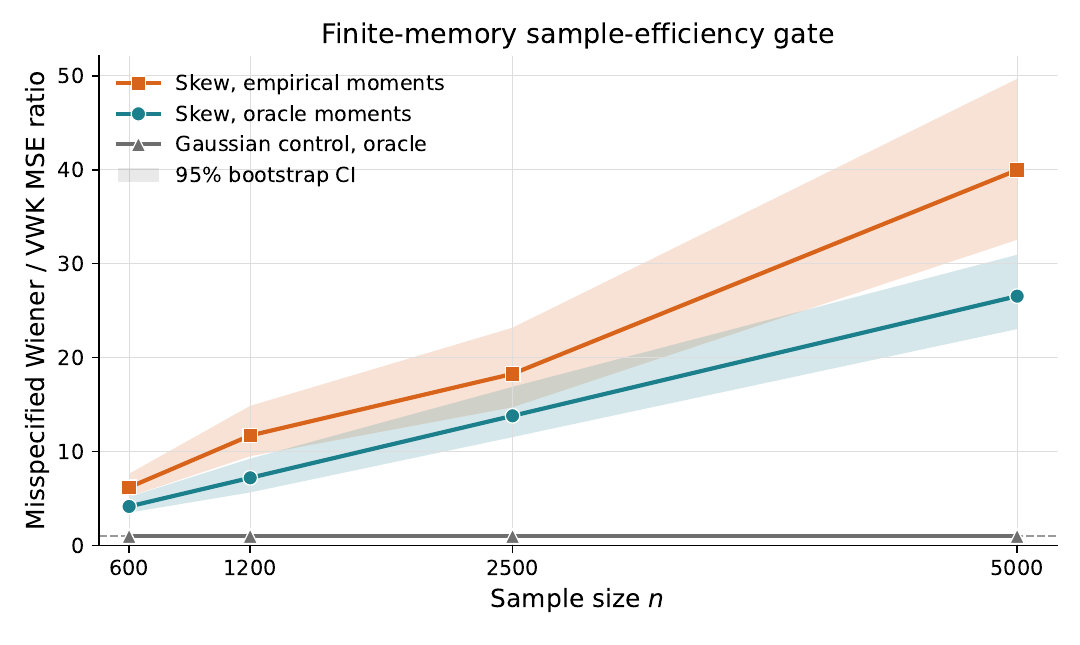}
\caption{Sample-efficiency curve for the finite-memory experiment. The skew
regime shows a persistent misspecified-Wiener penalty, whereas the Gaussian
control remains neutral for oracle moments. Shaded regions are $95\%$ bootstrap
confidence bands for the ratio of means; the Gaussian-control band collapses to
the line because its Wiener and VWK estimators coincide replicate by replicate.}
\label{fig:sample-efficiency}
\end{figure}

\FloatBarrier
\section{Real-world diagnostic screen}
\label{sec:realworld}

\paragraph{Problem setting.}
The controlled experiments above use known finite-memory polynomial signals. A
separate diagnostic question is whether the same projection effect can be seen
in ordinary regression data, where the true data-generating mechanism is unknown
and no Volterra kernel is available for recovery. To keep the question narrow,
we do not compare full predictors or tune model classes. For each prepared
scalar target we compare two coordinate-wise order-two projections over the same
one-driver polynomial envelope: the distribution-matched VWK projection and the
variance-matched Gaussian/Wiener projection. The diagnostic ratio is
$W/V$, the Wiener diagonal-projection MSE divided by the matched VWK
diagonal-projection MSE; values above one indicate a penalty for keeping the
diagonal estimator in the Gaussian basis.

\paragraph{Datasets and protocol.}
The targets are drawn from problem classes where non-Gaussian, heavy-tailed
regressors are the norm in signal processing: the SRU sulfur-recovery soft-sensor
benchmark \citep{Fortuna2007}, a canonical process-monitoring series; the UCI
gas-turbine CO/NO$_x$ emission-monitoring data \citep{Kaya2019GasTurbine}; the UCI
concrete compressive-strength data; and the French motor-insurance freMTPL2
severity data. Each screen row is a single scalar target regressed on a single
driver over the order-two envelope $\{1,x,x^2\}$; when a target admits several
candidate regressors the driver is fixed, before any VWK/Wiener comparison, as the
one minimizing the order-two training MSE. Target and driver are standardized using
\emph{training-fold} moments only, and the series is split $70/30$ into a leading
training fold and a held-out test fold with no shuffling, so that the
soft-sensor and emission series keep their temporal order. On the training fold we
fit four order-two estimators on the same design: the matched VWK diagonal
projection, the variance-matched Wiener diagonal projection, an unconstrained
full-span least-squares (LS) fit, and a Huber $M$-estimator ($\delta=1.345$). All
numbers reported are \emph{held-out test} MSE. A row is pre-registered as a
candidate illustration only when it clears two gates jointly: $W/V\ge1.2$
\emph{and} the matched projection stays within $25\%$ of the better guardrail,
$\mathrm{VWK}\le1.25\min(\mathrm{LS},\mathrm{Huber})$.

To avoid selective reporting, \cref{tab:realworld} lists all $11$ screened rows,
sorted by $W/V$; four clear both gates, one is the strongest-skew counterexample
($W/V<1$), and the remaining six are near-neutral ($0.97\le W/V\le1.05$). The
skew and excess-kurtosis columns are standardized moments of the training-fold
target residual after the order-two LS projection, not moments of the scalar
driver; for \texttt{fremtpl2\_severity\_raw} the skew $32.6$ is that of the
claim-amount target, while the driver \texttt{DrivAge} is the bounded regressor
defining the polynomial envelope.

\begin{sidewaystable}[t]
\centering
\caption{Complete real-world scalar screen, all $11$ prepared targets sorted by
$W/V$. Skew and excess kurtosis are standardized moments of the training-fold
target residual. The VWK, Wiener, and Huber columns are held-out test MSE; a
full-span least-squares fit reproduces the VWK column to machine precision (see
text) and is omitted. $W/V>1$ favors the matched VWK diagonal projection over the
variance-matched Wiener projection. Rows above the rule with $W/V>1.2$ that also
stay within $25\%$ of the Huber guardrail are candidate illustrations;
\texttt{fremtpl2\_severity\_raw} is the strongest-skew counterexample. These rows
are diagnostic only; the data carry no known Volterra kernels.}
\label{tab:realworld}
\scriptsize
\setlength{\tabcolsep}{2pt}
\begin{tabular}{@{}llrrrrrrr@{}}
\toprule
Dataset & Driver & $n$ & Skew & Exc.\ kurt. & VWK MSE & Wiener MSE & Huber MSE & $W/V$ \\
\midrule
\texttt{sru\_y1\_dynamic}           & \texttt{y1\_lag1} & 10079 & 10.557 & 274.212  & 0.133954 & 0.548277 & 0.132697 & 4.093 \\
\texttt{gas\_turbine\_co\_all\_raw} & \texttt{TIT}      & 36733 &  7.893 & 159.369  & 0.410362 & 1.058093 & 0.415731 & 2.578 \\
\texttt{gas\_turbine\_co\_2015\_raw}& \texttt{TIT}      &  7384 &  9.912 & 197.398  & 0.275787 & 0.454380 & 0.329211 & 1.648 \\
\texttt{sru\_y2\_dynamic}           & \texttt{y2\_lag1} & 10079 &  5.095 & 174.857  & 0.060072 & 0.087139 & 0.060441 & 1.451 \\
\texttt{sru\_y1\_static}            & \texttt{u2}       & 10081 &  7.553 &  95.701  & 1.404603 & 1.472416 & 1.267864 & 1.048 \\
\texttt{concrete}                   & \texttt{age}      &  1030 &  0.591 &  -0.226  & 0.545752 & 0.567897 & 0.523674 & 1.041 \\
\texttt{fremtpl2\_severity\_log}    & \texttt{Exposure} &  4000 & -0.392 &   3.195  & 0.828967 & 0.830935 & 0.829046 & 1.002 \\
\texttt{gas\_turbine\_nox\_all\_raw}& \texttt{AT}       & 36733 &  0.966 &   2.727  & 1.177836 & 1.166697 & 1.105880 & 0.991 \\
\texttt{gas\_turbine\_nox\_2015\_raw}& \texttt{AT}      &  7384 &  0.727 &   6.673  & 1.082969 & 1.064592 & 1.081131 & 0.983 \\
\texttt{sru\_y2\_static}            & \texttt{u3}       & 10081 &  3.019 &  29.072  & 0.955207 & 0.930168 & 0.910347 & 0.974 \\
\midrule
\texttt{fremtpl2\_severity\_raw}    & \texttt{DrivAge}  &  4000 & 32.613 & 1148.301 & 0.078410 & 0.074435 & 0.071228 & 0.949 \\
\bottomrule
\end{tabular}
\end{sidewaystable}

\paragraph{Guardrails: not underfitting, not merely robust regression.}
Two comparisons keep the positive rows honest. First, the matched VWK diagonal
projection reproduces the unconstrained full-span least-squares fit on
$\{1,x,x^2\}$ to machine precision on every row (test-MSE ratio $1.000$), which is
the finite-sample image of \cref{thm:mapping}: the diagonal estimator in the
matched orthonormal basis \emph{is} the least-squares projection onto the same
span. The VWK advantage over the Wiener column is therefore not an artifact of a
weaker or differently regularized fit\,---\,both estimators use identical training
data and the same order-two span, and VWK equals the optimal in-span fit. Second,
a Huber $M$-estimator answers the natural objection that any gain on heavy-tailed
data is simply robust regression: on the four flagged rows VWK matches or beats
Huber, tying it to within $1\%$ on \texttt{sru\_y1\_dynamic} and improving on it
by $1.3\%$ (\texttt{gas\_turbine\_co\_all\_raw}) up to $16\%$
(\texttt{gas\_turbine\_co\_2015\_raw}). The gain is thus a coordinate-geometry
effect, distinct from downweighting outliers.

\paragraph{Interpretation and scope.}
The screen gives a limited external check on the mechanism, not a claim of
field-data system identification. The positive rows show that a sizable
diagonal-projection gap can occur in non-synthetic regression targets:
\texttt{sru\_y1\_dynamic} reaches $W/V=4.093$, and the gas-turbine carbon-monoxide
targets give ratios $2.578$ and $1.648$. The six near-neutral rows and the
counterexample bound the claim. The most strongly skewed target,
\texttt{fremtpl2\_severity\_raw}, has residual skew $32.6$ but gives $W/V=0.949$,
so the matched diagonal projection does not win there; conversely the four
winners do not have the largest marginal skew. Thus residual skewness and
kurtosis are screening variables, not sufficient conditions for improvement. The
counterexample also marks the scope boundary the guardrails expose: on that
extreme insurance tail (residual skew $32.6$, excess kurtosis $1148$) the Huber
fit is the best of the four estimators, so where the second-moment geometry is
dominated by rare extreme claims the appropriate tool is robust $M$-estimation,
not a moment-matched projection.

This reading is consistent with \cref{prop:penalty} and \cref{subsec:deconfound}:
the diagonal-projection gap depends on the order-two moment geometry of the
projected signal, not on a marginal skewness statistic alone. Because these data
contain no known Volterra kernels, the screen is evidence that the proposed
misspecification mechanism can be observed outside synthetic examples; it is not
evidence for real-world kernel recovery or universal predictive dominance.

\FloatBarrier
\section{Discussion}
\label{sec:discussion}

The mathematical result is geometric. VWK replaces a fixed Gaussian coordinate
system by the orthogonal coordinate system generated by the actual input law.
This explains why the Gaussian control is neutral: when the Wiener assumption is
correct, the VWK basis and Wiener--Hermite basis coincide. It also explains why
the gain appears for asymmetric finite-moment input rather than for generic
non-Gaussianity, exactly as \cref{prop:penalty} predicts through
$\delta=\mu_3/\sigma^2$. In the centered-exponential regime, the order-two VWK basis
contains the skew correction needed to align the projection; in the symmetric
uniform and contaminated regimes used here, the same projection mismatch is not
observed.

The direct monomial least-squares and Huber baselines matter. Because monomial and
VWK coefficients span the same finite polynomial space, a full direct span fit can
perform very strongly in prediction error (\cref{subsec:deconfound}). VWK is
therefore not a universal prediction winner. Its role is to provide an orthogonal,
distribution-matched coordinate system in which the inexpensive diagonal
cross-correlation estimator is unbiased (\cref{prop:penalty}) and in which
regularization and finite-sample errors are decoupled from the ill-conditioned
power basis (\cref{subsec:conditioning}).

This reading also clarifies the connection to PMM and DSGE work in the broader
program. PMM2-like skew corrections are natural low-order consequences of
orthogonalization under an asymmetric law. DSGE-style reconstruction-error
features benefit from matched bases primarily when regularization interacts with
the conditioning of the basis. These are consequences of the geometry, not
separate claims of unconditional estimator dominance.

Several limitations are structural rather than merely numerical. First, the
proof assumes a product input law $P^d$; correlated inputs require
Gram--Schmidt in the joint lag law. Second, the theorem is finite-memory and
discrete; continuous-time Volterra kernels require additional measure-theoretic
work. Third, the present construction is moment-based; distributions without
finite moments require a separate characteristic-function treatment. Finally,
Lean currently verifies selected instances---low-order Hermite and the
arbitrary-$N$ Krawtchouk row---not the full VWK theorem.

\section{Conclusion}
\label{sec:conclusion}

This paper establishes a finite-memory moment-based VWK construction. Under finite
moments and monomial linear independence, oriented Gram--Schmidt in $L^2(P)$
gives the unique distribution-matched orthonormal polynomial basis. Tensor
products of this basis provide a finite-memory Volterra expansion, and monomial
Volterra coefficients and VWK coefficients are related by an invertible linear
change of basis. Wiener--Hermite appears as the Gaussian row, while classical
Askey families are closed instances for their corresponding laws.

The reproducible experiments support a narrow but strong empirical reading:
matched VWK projection removes a large misspecified Gaussian/Wiener projection
penalty in asymmetric finite-moment regimes, and this effect persists in a
second-order three-lag finite-memory benchmark. The same experiments do not
support a symmetric non-Gaussian advantage, a real-world kernel-recovery claim,
or universal superiority over direct robust span-fitting baselines.

The contribution is therefore not the distribution-matched basis, which is
arbitrary polynomial chaos, but the closed-form penalty for using the mismatched
Gaussian basis instead (\cref{prop:penalty}), the conditioning advantage that
follows (\cref{subsec:conditioning}), and a machine-checked instance of the
matched-basis recovery (\cref{prop:krawtchouk}). In this sense, VWK turns the
Wiener--Hermite cross-correlation estimator from a Gaussian-specific procedure
into a distribution-matched one with a quantified misspecification cost.

\bibliography{references}

\end{document}